\PassOptionsToPackage{
		pagebackref
	}{hyperref}

\documentclass[format=acmsmall, review=false, screen, nonacm]{acmart}
\usepackage{booktabs} 
\usepackage[ruled]{algorithm2e} 

\SetAlFnt{\small}
\SetAlCapFnt{\small}
\SetAlCapNameFnt{\small}
\SetAlCapHSkip{0pt}
\IncMargin{-\parindent}

\usepackage{amsmath}
\usepackage{amsthm}
\usepackage{microtype}
\usepackage{thmtools}
\usepackage{thm-restate}
\usepackage{natbib}
\usepackage{amsthm}
\usepackage[nameinlink]{cleveref}

\AtEndPreamble{%
	\theoremstyle{acmdefinition}

	}

\usepackage[nameinlink]{cleveref}
\usepackage{crossreftools}
\usepackage{subcaption}
\usepackage{tikz}
\usetikzlibrary{positioning}
\usepackage{xcolor}
\usepackage{nicefrac}
\usepackage{paralist}
\usepackage[colorinlistoftodos,textsize=footnotesize]{todonotes}
\tikzset{/tikz/notestyleraw/.append style={rounded corners=0pt,inner sep=0.6ex}}

\usepackage{graphicx} 
\usepackage{multirow} 

\usepackage{wrapfig}
\usepackage{etoolbox}
\makeatletter
\patchcmd\WF@putfigmaybe{\lower\intextsep}{}{}{\fail}%
\AddToHook{env/wrapfigure/begin}{\setlength{\intextsep}{0pt}}
\makeatother

\renewcommand{\epsilon}{\varepsilon}
\renewcommand{\le}{\leqslant}

\renewcommand{\ge}{\geqslant}

\usepackage{rotating}

\usepackage{tikz}
\usepackage{pgfplots}
\pgfplotsset{compat=1.15,
legend image code/.code={
\draw[mark repeat=2,mark phase=2]
plot coordinates {
(0cm,0cm)
(0.15cm,0cm)        
(0.3cm,0cm)         
};%
}}
\usepgfplotslibrary{groupplots}
\setcitestyle{authoryear}

\title[Independence of Approximate Clones]{Independence of Approximate Clones}

\author{Théo Delemazure}
\affiliation{%
	\institution{ILLC, University of Amsterdam}
	\country{Netherlands}
}
\email{theo.delemazure@uva.nl}

\begin{abstract}
	{\large Manuscript: January 2026}
	
	\bigskip
	\noindent
    In an ordinal election, two candidates are said to be perfect clones if every voter ranks them adjacently. The independence of clones axiom then states that removing one of the two clones should not change the election outcome. This axiom has been extensively studied in social choice theory, and several voting rules are known to satisfy it (such as IRV, Ranked Pairs and Schulze). However, perfect clones are unlikely to occur in practice, especially for political elections with many voters. 
    
    In this work, we study different notions of \emph{approximate clones} in ordinal elections. Informally, two candidates are approximate clones in a preference profile if they are close to being  \emph{perfect} clones. We discuss two measures to quantify this proximity, and we show under which conditions the voting rules that are known to be independent of clones are also independent of approximate clones. In particular, we show that for elections with at least four candidates, none of these rules are independent of approximate clones in the general case. However, we find a more positive result for the case of three candidates. Finally, we conduct an empirical study of approximate clones and independence of approximate clones based on three real-world datasets: votes in local Scottish elections, votes in mini-jury deliberations, and votes of judges in figure skating competitions. We find that approximate clones are common in some contexts, and that the closest two candidates are to being perfect clones, the less likely their removal is to change the election outcome, especially for voting rules that are independent of perfect clones.
\end{abstract}

\begin{document}

\begin{titlepage}
\maketitle
\vspace{10pt}
\hrule
\vspace{10pt}
\tableofcontents
\vspace{-18pt}
\hrule
\end{titlepage}

\section{Introduction}

A major issue in many electoral systems is what is called the \emph{spoiler effect}: one of the candidates that lost the election could have won  if another candidate (the spoiler) had not run  in the election. A typical example is the single-winner voting rule \emph{plurality}, in which every voter casts a vote for one candidate and the candidate with the highest score wins: if several candidates from the same side of the political spectrum are running, voters from this side are split between these candidates, which might harm their chances  to win. Sometimes, the spoiler effect can even be mutual and two candidates could each have won  individually if the other one had  withdrawn. 

To capture this effect, \citet{tideman1987independence} introduced the axiom of \emph{independence of clones} for elections based on ordinal preferences (i.e., in elections in which every voter casts a ranking of the candidates). In this ordinal model, we say that a set of candidates are \emph{clones} if every voter ranks them consecutively in their ranking. Then, the independence of clones axiom states that when such clones exist, removing all of them but one should not change the winner of the election (unless one of the removed clones is the winner, in which case the remaining clone should become the new winner). This axiom has been extensively studied and discussed in the social choice literature since then, particularly for single-winner and multi-winner voting, and has been adapted to other preference formats, such as approval preferences. With ordinal preferences, it has been shown that independence of clones is  satisfied by only a few single-winner voting rules (in particular Instant Runoff Voting, Ranked Pairs, and Schulze's rule) and multi-winner voting rules (in particular the Single Transferable Vote). Some of these rules are used in several countries for large-scale political elections. Notably, Instant Runoff Voting and Single Transferable Vote are used in Ireland, Australia, and some US states.

This being said, \emph{perfect} clones are unlikely to occur in any large-scale political election, as the number of voters is very high and in practice, some voters have unconventional preferences. Nevertheless, this axiom \emph{a priori} informs us about the expected behavior of a voting rule in presence of \emph{approximate} clones: if a rule is \emph{independent of clones}, then we can expect that it is less sensitive to the spoiler effect in general, in particular for approximate clones. But in practice, can we expect to find such \emph{approximate} clones? How should the proximity to being clones be defined for a pair of candidates? Moreover, do the rules that satisfy independence of clones also satisfy independence of approximate clones, if these candidates are close enough to being perfect clones? We aim to answer these questions in this work using two different notions of approximate clones, by studying the theoretical and empirical behavior of voting rules in presence of such approximate clones.

\subsection{Our contributions}

We first recall the model of ordinal preferences, the definition of clones and of the independence of clones axiom, and some single-winner voting rules satisfying this axiom in \Cref{sec:prelim}. Although clones can be sets of candidates of any size in the original definition, we focus on \emph{pairs} of candidates for simplicity in this work. However, it is clear that all negative results hold for larger sets of candidates, and our main positive result (\Cref{thm:irv-ranked-pairs-alpha-clones}) holds only for preference profiles with three candidates and thus only applies to pairs of approximate clones.

In \Cref{sec:approx}, we formally define the two notions of approximate clones that are used in this paper. These notions aim at quantifying how close a pair of candidates is to being clones. The first one, which we call $\alpha$-deletion clones, is based on the proportion $\alpha$ of voters that need to be removed from the profile for two candidates to become clones.\footnote{This notion is equivalent to the notion of \emph{MaxClones} already introduced by \citet{janeczko2024discovering}, and to the notion of \emph{independent clones} discussed by \citet{faliszewski2025identifying}.} The second one, which we call $\beta$-swap clones, is based on the number of swaps of adjacent candidates in voters' rankings that are required for the two candidates to be clones (here, $\beta$ is the average number of swaps required per voters). Note that $\alpha=\beta=0$ for perfect clones. 

In \Cref{sec:indep}, we show under which conditions and for which values of $\alpha$ and $\beta$ the voting rules that are known to be independent of perfect clones are also independent of approximate clones. Notably, in the general case and when there are four candidates or more, we show that none of these rules satisfy (weak) independence of approximate clones for any value of $\alpha > 0$ or $\beta >0$. However, we show that when there are three candidates, there are values of $\alpha$ for which they satisfy a weak version of the axiom.

Then, in \Cref{sec:real} we investigate approximate clones in real-world datasets of preferences. More specifically, we measure the proximity of the pairs of candidates to being clones (using our two notions), and study how different voting rules react to the deletion of approximate clones in practice. This analysis is based on three datasets: actual votes in local Scottish elections, rankings of judges in skating competitions, and votes of participants in deliberation experiments. We conclude and discuss potential future work in \Cref{sec:conclu}.

\subsection{Related Work}

One of the closest works to ours is the one of \citet{janeczko2024discovering}, who study among other problems the computational complexity of identifying approximate clones in preference profiles, using a notion equivalent to our $\alpha$-deletion clones measure. They also compute the proximity to being clones of pairs of candidates in synthetic preference profiles (from the map of elections \citep{szufa2020drawing}) and in real-world data. \citet{faliszewski2025identifying} continued this work by studying the parametrized complexity of this problem with respect to some natural parameters. They also consider another measure of proximity to being clones, based on the maximal distance between the two candidates in voters' rankings, which is different from the two notions we consider in this work.\footnote{All the negative results we obtain in this paper also hold for their measure (by using the same counter-examples), while our positive results do not hold.} The problem of identifying approximate clones has also been studied by \citet{DelemazureEtAl2026ApproximateClones} in the approval setting.

\citet{procaccia2025clone} also consider some notion of approximate clones for the problem of AI alignment, where the approximate clones are alternatives with similar features. In that sense, this work is also related to the literature on differential privacy \citep{dwork2006differential}. 
\citet{elkind2012clone} propose a richer structure than the notion of clones that captures clone sets of all possible sizes by means of PQ-tree. This notion was then used by \citet{berker2025independence} to define a distance between candidates which corresponds to the size of the smallest clone set containing both candidates, and which could be seen as another notion of approximate clones.

More generally, the independence of clones axiom has been studied in various contexts and has inspired a rich body of research in social choice theory. The notion, originally introduced in the context of single-winner voting with ordinal preferences \citep{tideman1987independence}, has later been generalized to weak orders \citep{schulze2011new,holliday2023split} and to approval preferences \citep{brandl2022approval}. Some variants have also been proposed and discussed for the multi-winner setting \citep{woodall1994properties,baumeister2024multiwinner}. Independence of clones is also related to the more demanding axiom of composition consistency \citep{laffond1996composition}, which additionally imposes which of the clones should win when a clone set is present in the election.

Finally, on the weakening of axioms through quantitative measures, \citet{bardal2025proportional} conduct a similar study of proportionality axioms for committee elections. 
\citet{delemazure2023selecting} introduce measures of how much conflict pairs of candidates induce, which can be seen as an opposite goal to what we are aiming for in this work (in particular, the \emph{discrepancy} measure they introduce can be linked to our notion of $\beta$-swap clones).

\section{Preliminaries} \label{sec:prelim}

Let $V = \{1,\dots,n\}$ be the set of voters and $C = \{c_1,\dots,c_m\}$ be the set of candidates. 
An ordinal preference profile is a collection $P = (\succ_1, \dots, \succ_n)$ where $\succ_i$ is a ranking (a strict total order) on $C$ representing the preferences of voter $i$ over the candidates. Given a ranking $\succ_i$, we denote by $\sigma_i(x)$ the position of candidate $x$ in $\succ_i$ (e.g., $\sigma_i(x) = 1$ if $x$ is ranked first). For a preference profile $P$ and a candidate $x \in C$, we denote by $P_{-x}$ the preference profile obtained by removing candidate $x$ from every ranking in $P$ (i.e.,  by projecting the profile onto $C \setminus \{x\}$).

A single-winner \emph{voting rule} $f$ is a function that takes as input an ordinal preference profile $P$ and returns a non-empty set of (tied) winners $f(P) \subseteq C$. We consider here \emph{irresolute} voting rules, i.e., rules that may return several winners in case of a tie.

In an ordinal profile $P$, we say that two candidates $x$ and $y$ are clones if they are ranked next to each other in every voter's ranking. More formally, a set $\{x,y\} \subseteq C$ with $x \ne y$ is a pair of clones if for every voter $i$ and $z \in C \setminus \{x,y\}$, we have $x \succ_i z$ if and only if $y \succ_i z$ (or equivalently, if $|\sigma_i(x) - \sigma_i(y)| = 1$ for every voter $i$). Note that this definition could be generalized to larger sets of clones, but we focus here on pairs of candidates for the reasons mentioned in the introduction.

Then, we say that a voting rule $f$ satisfies \emph{independence of clones} if adding or removing a clone of a candidate does not change (significantly) the result of the election. 
More formally, we have the following.

\begin{definition}
\label{def:independence-of-clones}
A voting rule $f$ satisfies \emph{independence of clones} if for every preference profile $P$ in which $x$ and $y$ are clones, we have (1) for every $z \in C \setminus \{x,y\}$, $z \in f(P)$ if and only if $z \in f(P_{-x})$ and (2)  $f(P) \cap \{x,y\} \ne \emptyset$ if and only if $y \in f(P_{-x})$.
\end{definition}

The first condition states that all non-clone winners remain winners after removing one of the clones, and that no new non-clone winner appears. The second condition states that if one of the clones is a winner, then the remaining clone should win after removing the other clone, and conversely, if none of the clones is a winner, then removing one of them should not make the other one win.

This axiom was first introduced by \citet{tideman1987independence}, who showed that the \emph{Instant Runoff Voting (IRV)} rule \citep{smith1973aggregation} satisfies this axiom, and introduced another rule called \emph{Ranked Pairs} that also satisfies it.\footnote{\citet{zavist1989complete} actually showed that in its original definition, Ranked Pairs failed independence of clones, and proposed a refinement of the rule that satisfies it. We refer to the work of \citet{berker2025independence} (Section 3.1) for more details.} Most of the other well-studied ordinal voting rules do not satisfy independence of clones. A notable exception is \emph{Schulze}'s rule \citep{schulze2011new}. We now define these rules more formally.

\begin{definition}
    \label{def:instant-runoff-voting}
    The \emph{Instant Runoff Voting} (IRV) rule  proceeds by eliminating candidates sequentially. At each round,  the candidate with the fewest first-place votes among non-eliminated candidates is eliminated. The last  remaining candidate is the winner.
\end{definition}

The next two rules (Ranked Pairs and Schulze) are based on the pairwise majority graph. In other words, their definitions rely on the notion of margin of victory, which is defined for a pair of candidates $(x,y) \in C^2$ as $M_{x,y} = n_{x \succ y} - n_{y \succ x}$ where $n_{x \succ y} = |\{i \in V \mid x \succ_i y\}|$ is the number of voters that prefer $x$ to $y$.

\begin{definition}
    \label{def:ranked-pairs}
    The \emph{Ranked Pairs} rule starts by sorting the pairs of candidates $(x,y) \in C^2$ by decreasing order of their margin of victory $M_{x,y}$. Then, the rule constructs a domination graph by adding the pairs in this order (if possible). For each pair $(x,y)$, it adds a directed edge from $x$ to $y$ as long as this does not create a cycle (in this case, the pair is skipped). The resulting graph defines a ranking of the candidates. The winner is the candidate at the top of this ranking.
\end{definition}

The next voting rule is based on the notion of \emph{path strength} in the pairwise majority graph. A path is a sequence of $k \ge 2$ candidates $\pi = (c_1, \dots, c_k)$. The \emph{strength} of the path is the minimum margin of victory along the path, i.e., $S_{\pi} = \min_{i=1}^{k-1} M_{c_i,c_{i+1}}$. 

\begin{definition}
    \label{def:schulze}
    The \emph{Schulze} rule is defined as follows. For every pair of candidates $(x,y) \in C^2$, we define the \emph{beatpath strength} $S_{x,y}$ to be the maximum strength over all paths from $x$ to $y$. Schulze's rule elects the candidate $x$ that beats every other candidate $y$ in beatpath strength, i.e., the candidate $x$ such that for every $y \in C \setminus \{x\}$, we have $S_{x,y} \ge S_{y,x}$. Such a candidate is guaranteed to exist.
\end{definition}

Note that these two rules are \emph{Condorcet rules}, i.e., they elect the Condorcet winner if it exists. As a reminder, a Condorcet winner $x \in C$ is a candidate that beats every other candidate in pairwise comparison, i.e., for every $y \in C \setminus \{x\}$, we have $M_{x,y} > 0$.

In case of ties at any step of these three rules, we consider all possible resolutions, and the set of winners is the union of the winners from all possible tie-broken outcomes. Note that this way of breaking ties is known as \emph{parallel-universe tie-breaking} \citep[Section 7]{conitzer2009preference}. However, in this work we will generally use examples without ties to simplify the presentation.

\section{Approximate Clones} \label{sec:approx}

We now introduce different notions of approximate clones, which quantify how close to being clones are two candidates in a given preference profile.

We say that two candidates $x$ and $y$ are \emph{$\alpha$-deletion clones} if we can remove at most $\alpha \cdot n$ voters from the election such that the remaining voters consider $x$ and $y$ to be clones. In other words, we can find a proportion $1-\alpha$ of voters for whom the two candidates are clones. More formally, we have the following.

\begin{definition}
    \label{def:alpha-clones}
    Given a preference profile $P$, we say that a pair of candidates $\{x,y\} \subseteq C$ with $x \ne y$ is a pair of \emph{$\alpha$-deletion clones} if:
    \[
    \alpha \ge \frac{1}{n}\left | \left \{i \in V \mid |\sigma_i(x) - \sigma_i(y)| > 1 \right \} \right |.
    \]
\end{definition}

For a pair of candidates, we generally consider the smallest value $\alpha$ for which they are $\alpha$-deletion clones. Note that this notion is equivalent to the notion of \emph{MaxClones} introduced by \citet{janeczko2024discovering}, and to the notion of \emph{independent clones} discussed by \citet{faliszewski2025identifying}.

One issue with this notion is that it focuses only on the voters for which the two candidates are clones, and ignores the voters for which they are not. In particular, the fact that two candidates are close in some rankings (without necessarily being adjacent) should also be taken into account. For instance, consider a first profile $P_1$ in which 70\% of voters rank $x$ and $y$ next to each other and the remaining 30\% rank them respectively first and last in their ranking (so with a rank difference of $m-1$), and a second profile $P_2$ in which 60\% of voters rank $x$ and $y$ next to each other, but the remaining 40\% rank them respectively first and third in their ranking (so with a rank difference of only 2). One could argue that $x$ and $y$ are closer to be clones in the second profile than in the first one, even though they are perfect clones for a larger share of voters in the first profile. This motivates the notion of \emph{$\beta$-swap clones}. Intuitively, a pair $\{x,y\}$ of candidates is a pair of $\beta$-swap clones if the total number of adjacent swaps required to make them perfect clones for all voters is at most  $\beta \cdot n$. More formally, we have the following.

\begin{definition}
    \label{def:beta-swap-clones}
    Given an ordinal preference profile $P$, we say that a pair of candidates $\{x,y\} \subseteq C$ with $x \ne y$ is a pair of \emph{$\beta$-swap clones} if:
    \begin{align*}
    \beta \ge \frac{1}{n}\sum_{i \in V} \left (|\sigma_i(x) - \sigma_i(y)| - 1\right ).
    \end{align*}
\end{definition}

Again, for a pair of candidates, we generally consider the smallest value $\beta$ for which they are $\beta$-swap clones. Note that when $\alpha = 0$ or $\beta = 0$, we recover the classical notion of perfect clones, and that when $m=3$, the two notions are equivalent (as in this case, we have $|\sigma_i(x) - \sigma_i(y)| - 1 \in \{0,1\}$ for all voters).

In what follows, we encompass both notions under the term \emph{approximate clones}. The following example illustrates how these two notions can differ in practice.

\begin{example}
    \label{ex:alpha-beta-clones}
    Consider the following preference profile $P$ with $n=10$ voters and $m=100$ candidates $\{a_1, \dots, a_{100}\}$:
        \begin{align*}
            5:~&  a_2 \succ a_1 \succ a_3 \succ \dots \succ a_{100} \\
            4:~& a_1 \succ a_2 \succ a_3 \succ \dots \succ a_{100}\\
            1:~& a_1 \succ a_3 \succ a_{100} \succ \dots \succ a_2 
        \end{align*}
        In this profile, candidates $a_1$ and $a_2$ are $0.1$-deletion clones, since only one voter out of ten does not rank them consecutively. However, this voter puts them very far away in their ranking, which makes them $9.8$-swap clones (98 swaps required in the last ranking). On the other hand, candidates $a_1$ and $a_3$ are $0.4$-deletion clones (4 voters do not rank them consecutively), but they are only $0.4$-swap clones (4 swaps required in total).
\end{example}

To get a better intuition of these two notions of approximate clones, let us now consider the case of the characteristic preference profiles that have been studied by \citet{faliszewski2023diversity}: the \emph{identity profile}, the \emph{antagonistic profile} and the \emph{uniform profile}.

In the \emph{identity profile}, all voters have the same ranking $\succ$, so two candidates are either $0$-deletion clones or $1$-deletion clones, depending on whether or not they are consecutive in $\succ$. The same is true in the \emph{antagonistic profile}, in which half of the voters have the same ranking $\succ$ and the other half have the reverse ranking $\overleftarrow{\succ}$. In both the identity and antagonistic profiles, the value of $\beta$ for which two candidates $x$ and $y$ are $\beta$-swap clones is equal to $|\sigma(x) - \sigma(y)| - 1$, where $\sigma$ is the position function associated to $\succ$.

However, in the \emph{uniform profile} in which every ranking appears exactly once (so the profile contains exactly $m!$ rankings), every pair of candidates is a pair of $\alpha$-deletion clones for the same value of $\alpha = (m-2)/m$. Similarly, every pair of candidates is a pair of $\beta$-swap clones for $\beta = (m-2)/3$. This can be deduced from the fact that the average values of $\alpha$ and $\beta$ over all pairs of candidates in any profile are respectively  $(m-2)/m$ and $(m-2)/3$, and by neutrality of the uniform profile, every pair of candidates has the same value of $\alpha$ and $\beta$.

\begin{proposition} \label{prop:average-alpha-beta}
    In any preference profile, the average value of $\alpha$ for which two candidates are $\alpha$-deletion clones over all pairs of candidates is equal to $(m-2)/m$, and the average value of $\beta$ for which two candidates are $\beta$-swap clones over all pairs of candidates is equal to $(m-2)/3$.
\end{proposition}

\begin{proof}
    Let us first prove that the average value of $\alpha$ for which two candidates are $\alpha$-deletion clones is $\alpha= (m-2)/m$.
    Let $k_{a,b}$ be the number of voters that rank $a$ and $b$ adjacently for $a\ne b$. For one voter, there exist exactly $m-1$ pairs of distinct candidates that are next to each other in their ranking. Thus, the sum of the $k_{a,b}$ over all voters is equal to $n\cdot (m-1)$ and the average value $\tilde{k} = \frac{2}{m(m-1)} \sum_{a,b} k_{a,b}$ over all pairs of candidates is equal to $\tilde{k} = \frac{2n(m-1)}{m(m-1)} = \frac{2n}{m}$. Therefore, the average value $\tilde{\alpha} = \frac{2}{m(m-1)} \sum_{a,b} \alpha_{a,b}$ (where $\alpha_{a,b}$ is the minimal value for which $a$ and $b$ are $\alpha$-deletion clones) is equal to $(n-\tilde{k}) / n = (m-2)/m$.
    
        Let us now prove that the average value $\tilde{\beta}$ for which two candidates are $\beta$-swap clones in any preference profile is equal to $(m-2)/3$. For a single voter, let us calculate the total number of swaps needed to make every pair of candidates adjacent. This is the sum of $|\sigma(x) - \sigma(y)| - 1$ over all pairs $\{x,y\}$ of distinct candidates. Consider a candidate at position $i$ and another at position $j < i$. They are separated by $i-j-1$ other candidates, so making them adjacent requires $i-j-1$ swaps. The total number of swaps for one voter is the sum of $(i-j-1)$ over all pairs of positions $\{i,j\}$ with $i>j$. This sum can be re-indexed and shown to be equivalent to:
    \begin{align*}
    \sum_{i = 3}^m \sum_{j=1}^{i-2} j 
    = \sum_{i = 1}^{m-2} \sum_{j=1}^{i} j 
    = \sum_{i = 1}^{m-2} \frac{i(i+1)}{2} = \frac{(m-2)(m-1)(m)}{6}.
    \end{align*}
    If we sum over all voters, we obtain $n(m-2)(m-1)(m)/6$. Because there are $m(m-1)/2$ pairs of candidates, the average number of swaps required to turn two candidates into clones is 
    \begin{align*}
    \frac{n(m-2)(m-1)(m)/6}{m(m-1)/2} = \frac{n(m-2)}{3},
    \end{align*}
    giving $\tilde{\beta} = (m-2)/3$.
\end{proof}

\Cref{prop:average-alpha-beta} directly implies that in every preference profile, there exists at least one pair of candidates that are $\alpha$-deletion clones for $\alpha \le (m-2)/m$, and one pair of candidates that are $\beta$-swap clones for $\beta \le (m-2)/3$. This means for instance that for $m=4$, there always exist candidates that are clones for at least half of the voters ($\alpha \le 1/2$). This also sets our expectations on what values of $\alpha$ and $\beta$ to expect in a perfectly random environment.

\section{Independence of Approximate Clones} \label{sec:indep}

Motivated by the idea that perfect clones rarely occur in practice, we are now interested in assessing whether voting rules that are independent of perfect clones are also independent of approximate clones. 
We say that a rule is independent of approximate clones (either $\alpha$-deletion clones or $\beta$-swap clones) if we can safely remove any of the approximate clones without changing the outcome of the election, similarly to the case of perfect clones.

\begin{definition}
    \label{def:indep-alpha-clones}
    We say that a voting rule $f$ is independent of $\alpha$-deletion clones (resp. of $\beta$-swap clones) if for each possible profile $P$ and two candidates $x,y\in C$ that are $\alpha$-deletion clones (resp. $\beta$-swap clones) in $P$, we have (1) for all $z \in C \setminus \{x,y\}$, $z \in f(P)$ if and only if $z \in f(P_{-x})$, and (2) $f(P) \cap \{x,y\} \ne \emptyset$ if and only if $y \in f(P_{-x})$.
\end{definition}

Note that the roles of $x$ and $y$ are symmetric in this definition. Independence of $0$-deletion clones (and of $0$-swap clones) corresponds to the classical axiom of independence of clones. On the other hand, independence of $1$-deletion clones is  much stronger than the well-known (and already very strong) axiom of \emph{independence of losers}, which is the single-winner voting version of the \emph{independence of irrelevant alternatives} axiom \citep{arrow1950difficulty} and which states that the outcome of the election should not change when we remove any candidate that is not a winner.\footnote{Note that for $\alpha=1$, any pair of candidates satisfies the condition, meaning if we remove any winning candidate, all other candidates must become winners.}

Unfortunately, a simple example is enough to show that any reasonable rule fails this axiom for any $\alpha > 0$ or $\beta > 0$.

\begin{theorem} \label{thm:indep-alpha-clones-imp}
    Any voting rule that coincides with the majority rule when $m=2$ fails independence of $\alpha$-deletion clones for any $\alpha> 0$ and independence of $\beta$-swap clones for any $\beta > 0$.
\end{theorem}

\begin{proof}
    Let $f$ be a rule satisfying the majority rule when $m=2$, and which is independent of $\alpha$-deletion clones for some $\alpha > 0$. Consider the following profile $P$:
    \begin{align*}
        k:~& a \succ b \succ c&
        k:~&c \succ b \succ a& 
        1:~ &a \succ c \succ b
    \end{align*}
    
    In this profile, $b$ and $c$ are perfect clones (and thus $\alpha$-deletion clones) and $a$ and $b$ are $1/(2k+1)$-clones, which tends to 0 as $k$ tends to infinity. Thus, for any $\alpha > 0$, there exists $k$ such that $1/(2k+1) < \alpha$. 
    Let us assume that $a \in f(P)$. By independence of $\alpha$-deletion clones applied to the pair $\{a,b\}$, we should have $b \in f(P_{-a})$. However, $c$ wins the majority vote against $b$ in $P_{-a}$, a contradiction. If $b \in f(P)$, by independence of $\alpha$-deletion clones applied to the pair $\{b,c\}$,we should have $c \in f(P_{-b})$, which again contradicts the majority vote. Finally, if $c \in f(P)$, we use the same argument by reversing the roles of $b$ and $c$.

    Since both measures of approximate clones are equivalent when $m=3$, the same proof works for $\beta$-swap clones.
\end{proof}

In particular, IRV, Ranked Pairs and Schulze's rule fail this property. It is not too surprising that we obtain such a negative result: the voters for whom the two candidates are not clones might put one of the clones first and the other one last (like in the example from the proof above). In this case, if the first clone (the one ranked first) is the winner, it may be too demanding to ask that when we remove it, the other clone (who is ranked last by some voters) should replace it.  

\subsection{Weak independence of approximate clones}

A weaker and perhaps easier to satisfy version of this axiom is to ask that at least one of the two approximate clones can be removed without changing the outcome of the election. Here, only one of the clones is treated as the potential spoiler, and its presence (or absence) should not impact the result.

\begin{definition}
    We say that a voting rule $f$ is \emph{weakly} independent of $\alpha$-deletion clones (resp. of $\beta$-swap clones) if for any profile $P$ and two candidates $x,y\in C$ that are $\alpha$-deletion clones (resp. $\beta$-swap clones) in $P$, the following holds for \emph{either} $P' = P_{-x}$ or $P' = P_{-y}$: (1) for all $z \in C \setminus \{x,y\}$, $z \in f(P)$ if and only if $z \in f(P')$ and (2) $f(P) \cap \{x,y\} \ne \emptyset$ if and only if $f(P') \cap \{x,y\} \ne \emptyset$.
\end{definition}

Clearly, independence of $\alpha$-deletion clones implies weak independence of $\alpha$-deletion clones. One could think of intermediate axioms, but we will see that even our weak notion leads to quite negative results, which reduce the need for further strengthening.

Note that weak independence of clones prevents a mutual spoiler effect (i.e., cases in which each of the two clones wins when running alone, but both clones lose when they both run) but it does not prevent all kinds of spoiler effect. For instance, weak independence allows for a one-way spoiler effect: It could be that  $x$ wins when running alone, but both $x$ and $y$ lose when running together, and the axiom would still be satisfied as long as $y$ also loses when running alone (as in the counterexample in the proof of \Cref{thm:indep-alpha-clones-imp}). Unfortunately, this phenomenon can happen with IRV, Ranked Pairs and Schulze's rule for any $\alpha > 0$.

In the remainder of this section, we will assess whether some rules that are known to be independent of (perfect) clones are also weakly independent of $\alpha$-deletion clones or $\beta$-swap clones. We can first show that when $m \ge 4$, they are not weakly independent of $\alpha$-deletion clones and of $\beta$-swap clones for any $\alpha > 0$ or $\beta > 0$.

\begin{theorem}
\label{thm:irv-alpha-clones}
    For $m \ge 4$, IRV, Ranked Pairs and Schulze are not weakly independent of $\alpha$-deletion clones (resp. $\beta$-swap clones) for any $\alpha > 0$ (resp. $\beta > 0$).
\end{theorem}

\begin{proof}
    Let $\alpha > 0$ and $k \in \mathbb N$. We start with IRV, and consider the following profile $P$:
    \begin{align*}
        k:~ & a \succ b \succ c \succ d &
        k:~  & b \succ a \succ c \succ d &
        k:~  & d \succ a \succ b \succ c \\
        k+1:~  & c \succ d \succ a \succ b &
        2:~  & a \succ d \succ b \succ c &
        2:~  & b \succ d \succ a \succ c
    \end{align*}
    In this profile, $a$ and $b$ are $4/(4k+5)$-deletion clones, which tends to $0$ when $k$ tends to infinity. Thus, there exists $k$ such that $\alpha > 4/(4k+5)$. With IRV, $d$ is eliminated first, then $c$, then $b$, and finally $a$ wins. By independence of $\alpha$-deletion clones, $a$ should win when we remove $b$ or $b$ when we remove $a$. However, if we remove $a$, then $c$ is eliminated first, giving $k+1$ points to $d$, and $b$ is eliminated next, thus $d$ wins. Similarly, if we remove $b$, $d$ wins. This contradicts weak independence of $\alpha$-deletion clones.

    For Ranked Pairs, consider the following profile $P$:
    \begin{align*}
        k+7:~ & d \succ a \succ b \succ c &
        k:~ & c \succ a \succ b \succ d &
        10:~  & b\succ c \succ a \succ d \\
        2:~  & a \succ b \succ d \succ c &
        3:~ & c \succ a \succ d \succ b & 
        4:~& d \succ c \succ a \succ b 
    \end{align*}
    In this profile, $a$ and $b$ are $13/(2k+26)$-clones, which tends to $0$ when $k$ tends to infinity. Therefore, there exists $k$ such that $\alpha > 13/(2k+26)$. The margins of victory are as follows:
\begin{center}
        \begin{tabular}{c| c c c c}
        & $a$ & $b$ & $c$ & $d$ \\
        \hline
        $a$ & 0 & $2k+6$ & $-8$ & $4$ \\
        $b$ & $-(2k+6)$ & 0 & $12$ & $-2$ \\
        $c$ & $8$ & $-12$ & 0 & $0$ \\ 
        $d$ & $-4$ & $2$ & $0$ & 0         
    \end{tabular}
\end{center}

    The pairs are added in the following order: first $a \succ b$, then $b \succ c$. Then, $c \succ a$ is skipped because it would create a cycle. Then $a \succ d$ and $d \succ b$ are added. Thus, the final ranking is $a \succ d \succ b \succ c$ and the winner is $a$. Now, if we remove $a$, the pairs $b \succ c$ and $d \succ b$ are added, and the winner is $d$. If we remove $b$, the pairs $a \succ d$ and $c \succ a$ are added and the winner is $c$. Thus, a clone wins when both clones are in the election, but not when we remove one, which contradicts weak independence of $\alpha$-deletion clones.

    Using the same profile, we can also show that Schulze's rule is not weakly independent of $\alpha$-deletion clones for $\alpha > 0$. 
    Indeed, $a$ beats every other candidate in terms of beatpath strength, and wins the election. However, if we remove $b$, then $c$ beats $a$ and $d$ in terms of beatpath strength, and wins, and if we remove $a$, $d$ beats $b$ and $c$ in terms of beatpath strength, and wins. Again, this contradicts weak independence of $\alpha$-deletion clones.

    In the profiles above, $a$ and $b$ have the same value of $\beta$-swap clones as of $\alpha$-deletion clones, so the same proof works to show that these rules fail independence of $\beta$-swap clones for $\beta > 0$.
\end{proof}

An open question is now: is there a (sensible) ordinal voting rule that is independent of approximate clones for some $\alpha > 0$ or $\beta > 0$ when $m \ge 4$? 

The situation is different when there are only $m= 3$ candidates. In this case, IRV is weakly independent of $\alpha$-deletion clones for $\alpha \le 1/3$, and Ranked Pairs and Schulze for any $\alpha \ge 0$. Remember that when $m=3$, the two measures are equivalent, thus we have the same results for $\beta$-swap clones. However, the result is only true if we restrict the possible profiles to those for which the rule is resolute (i.e., produces a single winner). For $f$ a voting rule, we call $f$-simple profiles the profiles for which $f$ produces a single winner.

\begin{theorem}
\label{thm:irv-ranked-pairs-alpha-clones}
    For $m=3$, IRV is weakly independent of $\alpha$-deletion clones and of $\beta$-swap clones on IRV-simple profiles for $\alpha \le 1/3$, and not for $\alpha > 1/3$. Ranked Pairs and Schulze are weakly independent of $\alpha$-deletion clones on Ranked Pairs- (resp. Schulze-) simple profiles for any $\alpha \ge 0$.
\end{theorem}

\begin{proof}
Let us first show that IRV is not weakly independent of $\alpha$-deletion clones for $\alpha > 1/3$. Consider the following profile with 3 candidates $a,b,c$:
\begin{align*}
    2k:~ &c \succ a \succ b &
    k+1:~ & a \succ c \succ b &
    k+1:~ & b \succ c \succ a \\
    k:~ & a \succ b \succ c &
    k:~ & b \succ a \succ c
\end{align*} 

In this profile, $a$ and $b$ are $(2k+2)/(6k+2)$-clones, which tends to $1/3$ when $k$ tends to infinity. Thus, there exists $k$ such that $(2k+2)/(6k+2) < \alpha$ for any $\alpha > 1/3$. With IRV, $c$ is eliminated first, then $a$ wins. However, if we remove either $a$ or $b$, $c$ wins, which contradicts weak independence of $\alpha$-deletion clones.

Now, let us prove that IRV is independent of $\alpha$-deletion clones for $\alpha \le 1/3$. If in the profile with both clones, one of the clones is eliminated first, then removing it from the election will not change the outcome. Thus, the only elections for which removing either one of the clones can change the result is when the first eliminated candidate is not a clone. In this case, a breach of independence of approximate clones implies that removing either of the clones should make this third candidate win, like in the example above. Let us assume without loss of generality that $a$ and $b$ are $\alpha$-deletion clones, and that $c$ is the first eliminated candidate. There are only 6 possible rankings of the three candidates. Because $c$ is eliminated first, its plurality score is necessarily lower or equal than $n/3$. Therefore, there are at most $n/3$ rankings of the form $c \succ a \succ b$ or $c \succ b \succ a$. Let us denote $n_1 \le n/3$ the number of such rankings. Moreover, since $a$ and $b$ are $1/3$-deletion clones, there are at most $n/3$ rankings of the form $a \succ c \succ b$ or $b \succ c \succ a$. Let us denote $n_2$ the number of such rankings. Combining these two facts, we obtain that there are \emph{at least} $n - n_1 - n_2 \ge n/3$ rankings of the form $a \succ b \succ c$ or $b \succ a \succ c$. Moreover, since there are $n_2$ rankings of the form $a \succ c \succ b$ or $b \succ c \succ a$, one of the two rankings appears more than $n_2/2$ times by pigeonhole principle. Let us assume without loss of generality that $a \succ c \succ b$ appears more than $n_2/2$ times. Then, if we remove $b$, the number of rankings of the form $a \succ c$ is greater than $n_2/2 + (n - n_1 - n_2) = n - n_1 - n_2/2 \ge n - n/3 -n/6 = n/2$, which is enough to make $a$ win. This proves that IRV is independent of $\alpha$-deletion clones for $\alpha \le 1/3$.

Let us now consider Ranked Pairs. Clearly, if there are no cycles involving the three candidates in the pairwise majority graph, then the outcome will not change when we remove a losing approximate clone. Let us assume that there is a cycle between the three candidates $a,b$ and $c$ in the pairwise majority graph. Without loss of generality, assume that $a$ and $b$ are $\alpha$-deletion clones and the cycle is $a \succ b \succ c \succ a$. If $c$ wins in this profile, it still wins when we remove $b$, and if either $a$ or $b$ wins, then $b$ wins when we remove $a$. In both cases, the axiom is satisfied. Since Ranked Pairs and Schulze coincide when $m=3$, the above proof also works for Schulze.
\end{proof}

Without the simple-profile condition, the profile $P = \{1: a \succ c \succ b, k:a \succ b \succ c, k+1:b \succ a \succ c, 2k: c \succ a \succ b\}$ is a counterexample for IRV for arbitrary small $\alpha > 0$ (by choosing $k$ accordingly). In this profile, $a$ and $b$ are approximate clones, and in the first step both $a$ and $b$ can be eliminated, thus giving IRV$(P) = \{a,c\}$ (using parallel-universe tie-breaking). However, if we remove $b$, only $a$ wins, and if we remove $a$, only $c$ wins, violating weak independence of $\alpha$-deletion clones in both cases.

Note that the fact that Ranked Pairs and Schulze are weakly independent of $\alpha$-deletion clones for any $\alpha \le 1$ when $m=3$ can be seen closer to a result on independence of losers than to a result on independence of approximate clones.

\subsection{Subclasses of profiles}

Let us now consider subclasses of profiles for which rules are weakly independent of $\alpha$-deletion clones and $\beta$-swap clones. Clearly, any Condorcet rule is weakly independent of $\alpha$-deletion clones and of $\beta$-swap clones for any $\alpha \ge 0$ or $\beta \ge 0$ when there exists a Condorcet winner, since at most one of the approximate clones can be the winner, and removing a loser does not change the Condorcet winner. This result includes single-peaked profiles, which necessarily have a Condorcet winner \citep{black1948rationale}.

We can extend this result using the notion of \emph{Smith set}, which is the set of candidates of minimum size such that every candidate inside the set beats every candidate outside the set in pairwise comparison. More formally, it is the set $S$ of minimum size such that for all $x \in S$ and $y \in C \setminus S$, we have $M_{x,y} > 0$. If the Smith set contains exactly one candidate, then this candidate is the Condorcet winner. A voting rule is said to be \emph{independent of Smith-dominated alternatives} if removing a candidate that is not in the Smith set does not change the outcome of the election. Additionally, a voting rule is said to satisfy the \emph{Smith criterion} if the winner is necessarily from the Smith set. Both axioms are satisfied by Ranked Pairs and Schulze's rule, as well as Copeland \citep{coleman1966possibility} and Kemeny-Young \citep{kemeny1959mathematics}. We can show the following result. 

\begin{theorem}
\label{thm:smith-alpha-clones}
    If $n$ is odd, any voting rule that satisfies the Smith criterion and is independent of Smith-dominated alternatives is weakly independent of $\alpha$-deletion clones (resp. $\beta$-swap clones) for any $\alpha \ge 0$ (resp. $\beta \ge 0$) when the size of the Smith set is at most 3.
\end{theorem}

\begin{proof}
     The proof follows the same idea as the proof of \Cref{thm:irv-ranked-pairs-alpha-clones}. If the Smith set contains only one candidate, then the result is clear, because at least one of the approximate clones must be outside the Smith set. By independence of Smith-dominated alternatives, removing this non-Smith candidate does not change the winner set, thus satisfying weak independence. 
     
     Now, assume that the Smith set contains two candidates. If at least one of the approximate clones is not in the Smith set, we can remove it without changing the outcome of the election by independence of Smith-dominated alternatives. If both candidates in the Smith set are clones, then removing either one of them does not change the outcome of the election, because the Smith set only contains the two clones.

     Finally, consider the case where the Smith set contains three candidates. If one clone is not in the Smith set, then we can safely remove it by independence of Smith-dominated alternatives. If both clones are in the Smith set, then there is a cycle between the three candidates, and the third candidate is dominated by one clone and dominates the other. Note that since $n$ is odd, pairwise majority margins are such that $M_{x,y} \ne 0$ for all $x,y \in C$. If one of the clones wins in the original profile, removing the dominated clone makes the other clone win by the Smith criterion (because it becomes the Condorcet winner). If the non-clone candidate of the Smith set wins in the original profile, then removing the dominating clone makes the non-clone candidate win by the Smith criterion (because it becomes the Condorcet winner). Thus, in all cases, one of the approximate clones can be removed such that the conditions of weak independence of approximate clones are met.
\end{proof}

Note that the proof also applies for even $n$ when there are no ties in the pairwise majority graph (i.e., for all $x,y \in C$, $M_{x,y} \ne 0$).

This result is particularly relevant, since it is known that in practice, when the number of voters is large, most preference profiles admit a Condorcet winner, and profiles that do not admit a Condorcet winner generally have a Smith set of size at most 3 \citep{graham2025examination,tideman2011modeling}. However, this result holds for \emph{weak} independence of $\alpha$-deletion clones, but not anymore for independence of $\alpha$-deletion clones (in particular, the profile used in the proof of \Cref{thm:indep-alpha-clones-imp} contains a Condorcet winner).

\section{Empirical Analysis} \label{sec:real}

Let us now investigate approximate clones and whether rules are independent of approximate clones in real-world datasets of preferences. We will study the following datasets:

\begin{itemize}
    \item \textbf{Local Scottish elections}: This dataset was collected by \citet{mccune2024monotonicity} and analyzed by \citet[Section 3]{bardal2025proportional}. These are votes in multi-winner elections in Scotland, in which the Single Transferable Vote (STV) rule is used, which is the equivalent of IRV for committee voting and is known to be independent of clones \citep{tideman1995single}. Because these are multi-winner elections, in many instances several candidates from the same party are running. There is a total of 1\,070 preference profiles, with between $m=3$ and $14$ candidates and between $n=664$ and $14\,217$ voters. Voters were not required to cast full rankings and could give truncated rankings. However, we only kept the votes with a full ranking, giving us between $n=43$ and $2\,905$ voters per profile. The distributions of the number of voters and candidates in this dataset are shown in \Cref{fig:scottish_stats}.
    \item \textbf{Figure skating competitions}: This dataset was collected by \citet{smith2000analysis} (based on numerical scores given by judges to participants in figure skating competitions) and is available on Preflib \citep{MaWa13a}. There are a total of 48 profiles, with between $n=7$ and $9$ voters (judges) and between $m=14$ and $30$ candidates. 
    \item \textbf{Mini-jury deliberation}: This last dataset was collected by \citet{tessler2024ai} to train an AI system to mediate human deliberation. It contains 2\,581 profiles of $n=5$ voters and $m=4$ candidates, in which the voters were asked to rank several group statements on various topics (e.g., climate change, universal basic income, etc.),  where statements were provided by an LLM as an aggregate of voters' personal opinions on the topic discussed. The Schulze rule was used to aggregate the rankings.
\end{itemize}

\begin{figure}[!t]
 \centering
    \begin{subfigure}{.49\textwidth}
 \includegraphics[width=1\textwidth]{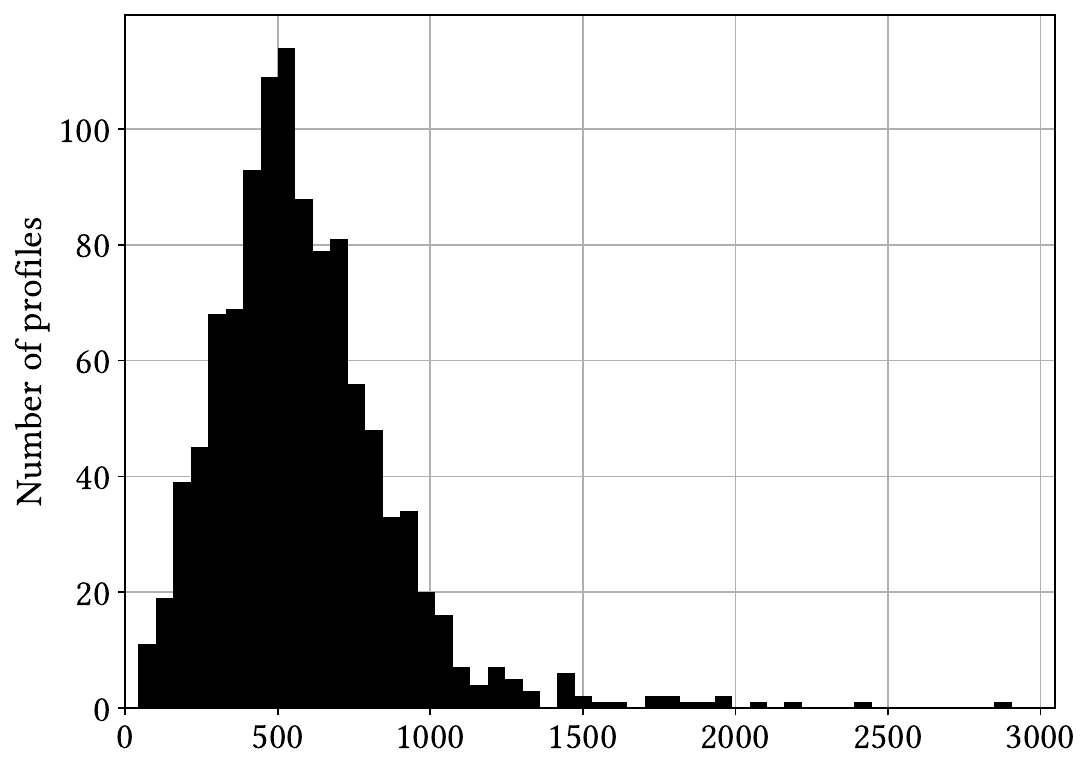}
 \caption{Distribution of the number of voters $n$.}
 \label{fig:scottish_n}
    \end{subfigure}
    \begin{subfigure}{.49\textwidth}
 \includegraphics[width=1\textwidth]{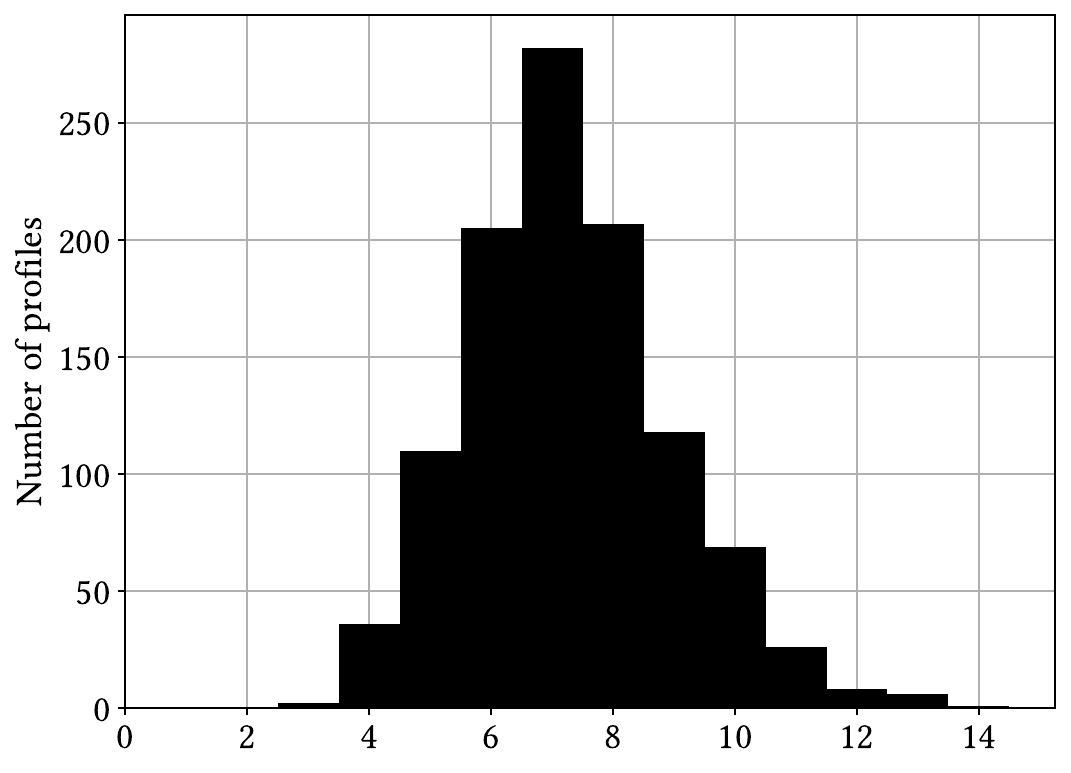}
 \caption{Distribution of the number of candidates $m$.}
 \label{fig:scottish_m}
    \end{subfigure}
    \caption{Distribution of the number of voters (who ranked all candidates) and of candidates in the Scottish elections dataset.}
    \label{fig:scottish_stats}
 \end{figure}

Each of these datasets is likely to contain approximate clones. In the Scottish elections, parties often presented several candidates in the same election, which are potential approximate clones, as they are likely to be ranked similarly by voters. In figure skating competitions, rankings of different judges are often very similar, since they are mostly based on participants' performances and skills, which can be seen as an objective measure. In that sense, participants with similar skills are likely to be perceived as approximate clones. Finally, in the mini-jury deliberation dataset, statements are often very similar, as they are all generated with the same input (voters' personal opinions). Note that in both the Scottish elections and the deliberation votes, the voting rules used (STV and Schulze) are known to be independent of perfect clones. 

\subsection{Approximate Clones}

We first investigate the presence of approximate clones in these datasets. For each profile, we computed for each pair of candidates the (minimal) value of $\alpha$ and $\beta$ for which they are $\alpha$-deletion clones and $\beta$-swap clones. We are particularly interested in the minimum value of $\alpha$ and $\beta$ over all pairs of candidates in each profile, as this corresponds to the pair of candidates that are the closest to being clones.

 \begin{figure}[!t]
 \centering
    \begin{subfigure}{.49\textwidth}
\includegraphics[width=\textwidth]{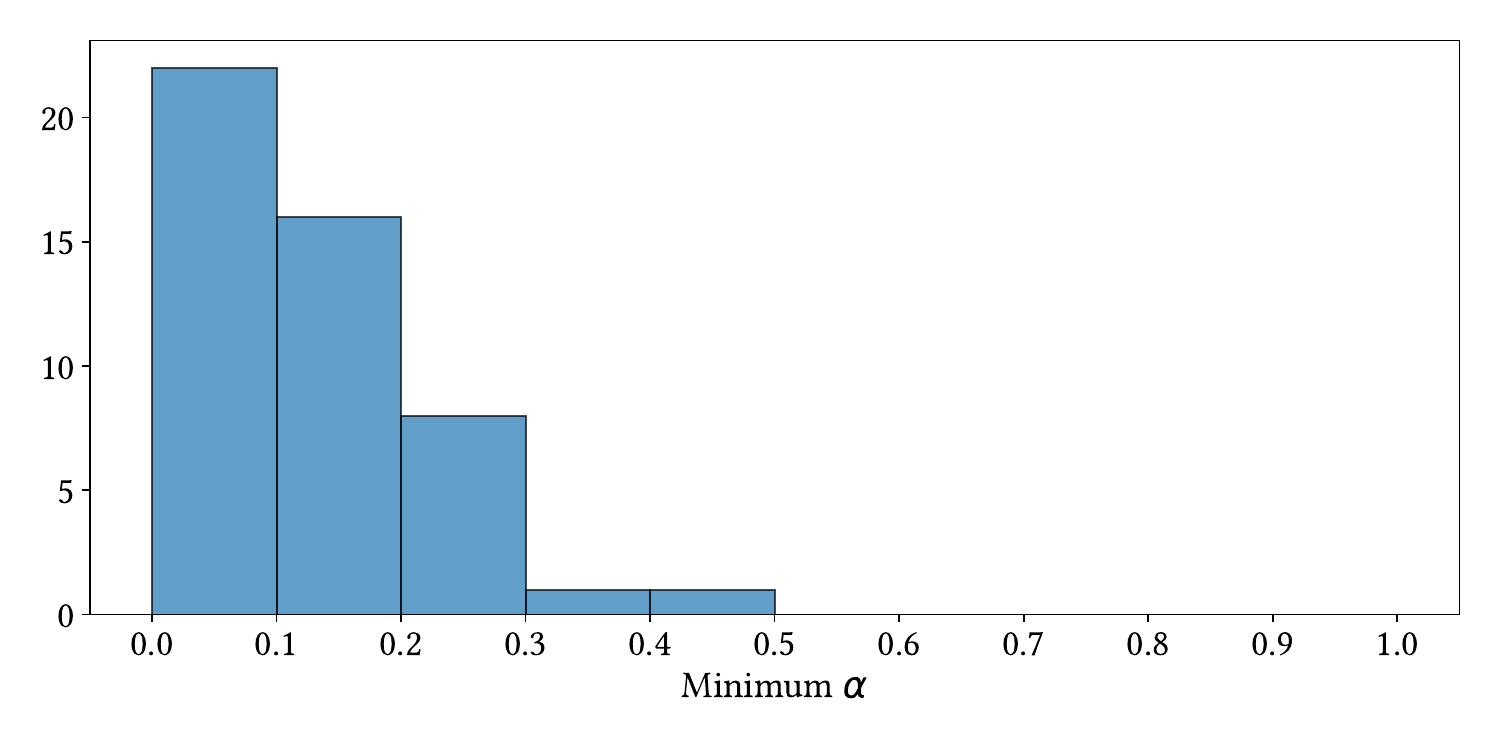}
    \end{subfigure}
    \begin{subfigure}{.49\textwidth}
\includegraphics[width=\textwidth]{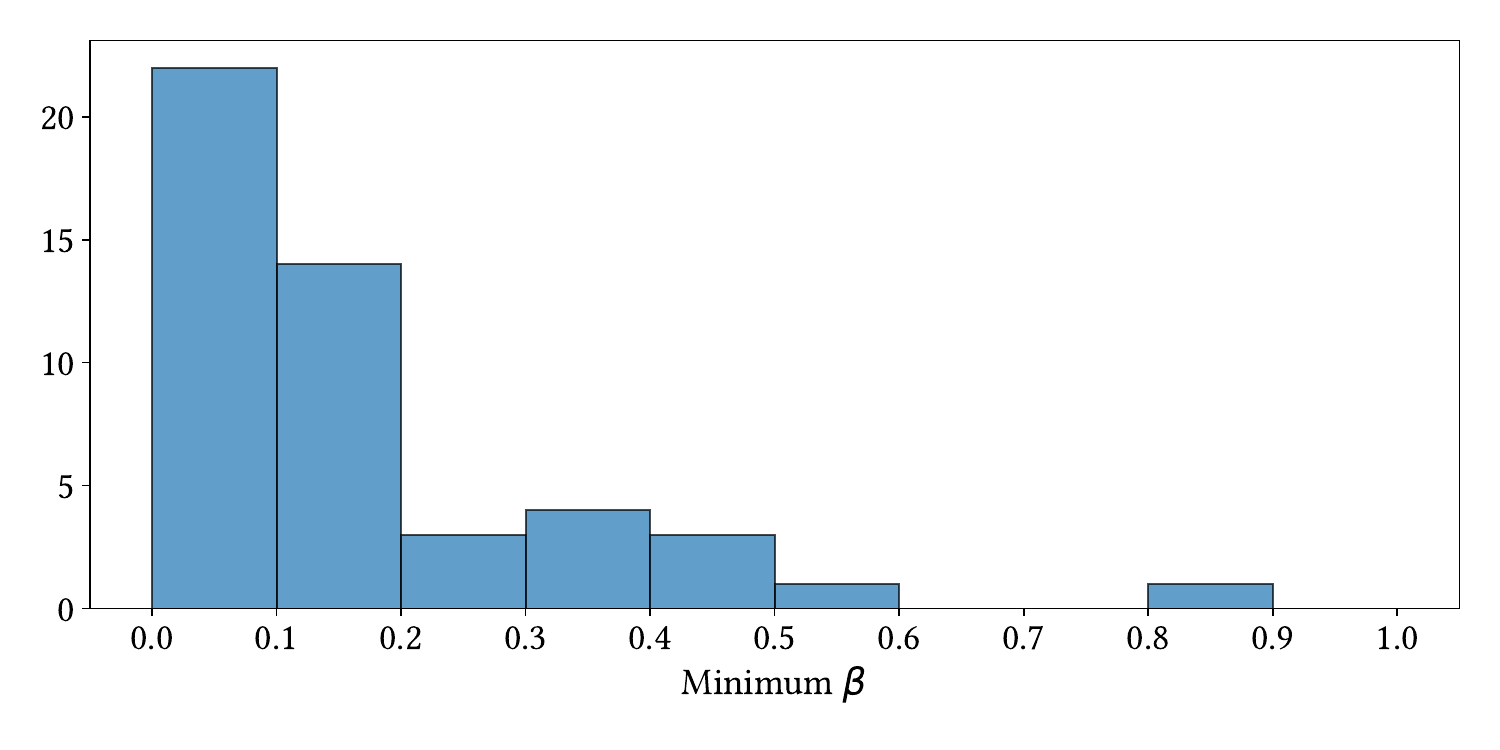}
    \end{subfigure}
\caption{Distributions (in number of instances) of the minimum value of $\alpha$ (left) and $\beta$ (right) for which pairs of candidates  are respectively $\alpha$-deletion clones and $\beta$-swap clones in the figure skating competitions dataset.}
\label{fig:skating_distributions}
\end{figure}

The dataset which contains the highest proportion of perfect and approximate clones is the figure skating competition, with 46\% of profiles (22 out of 48) containing perfect clones, and average minimum values of $\alpha$ and $\beta$ of respectively $0.09$ and $0.135$ over all profiles (see \Cref{fig:skating_distributions} for the full distributions). This is mostly because the different judges cast very similar rankings, and because the number of voters is quite low. Because of this, we also observe pairs of candidates with very high values of $\alpha$ (meaning only a few voters rank them consecutively) but low values of $\beta$ (meaning that they are close to each other in voters' ranking). This is for instance the case if every jury has the same ranking and we are looking at the candidates in first and third position. Thus, because of the particular structure of this dataset, the two measures are not always correlated and the $\beta$ value is much more informative.

\begin{figure}[!t]
\centering
\includegraphics[width=\textwidth]{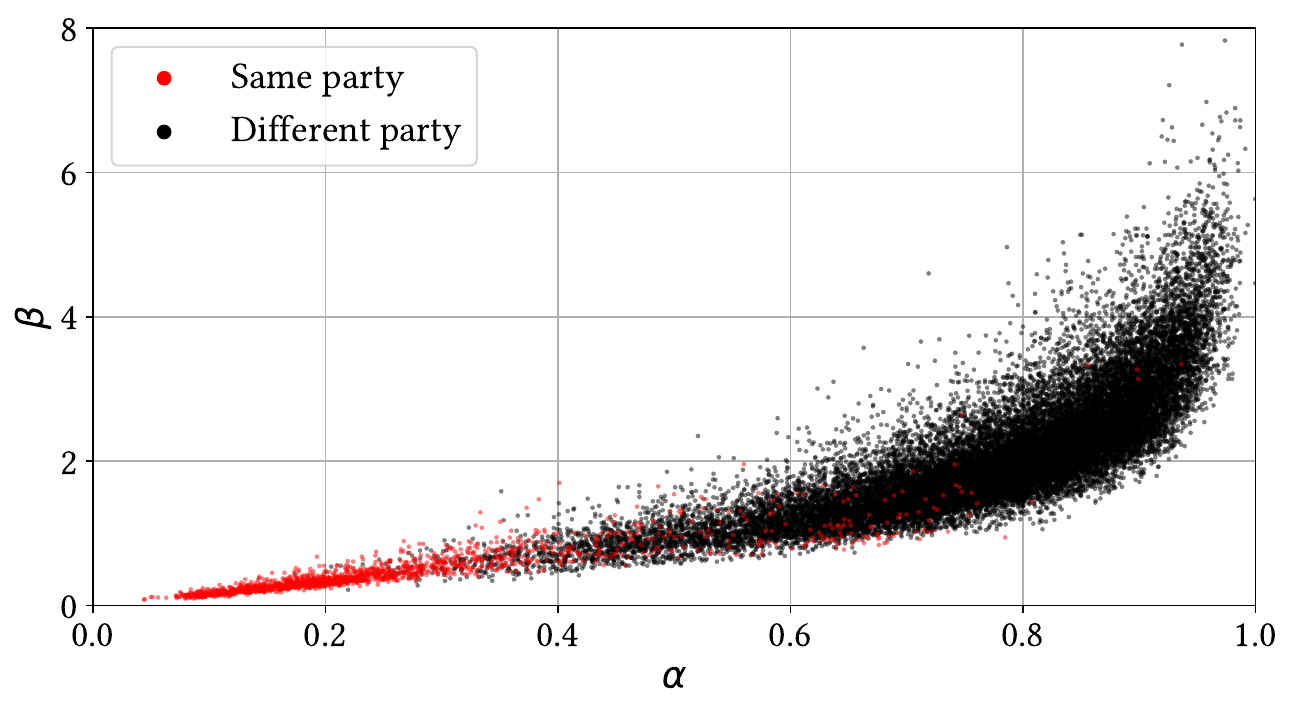}
\caption{Values of $\alpha$ and $\beta$ for all pairs of candidates in the Scottish elections dataset. Red dots correspond to pairs of candidates from the same party, and black dots correspond to pairs from different parties.}
\label{fig:scotland_1}
\end{figure}

If we now turn to the Scottish elections dataset, we observe that none of the preference profiles admit perfect clones, despite the presence of candidates from the same parties. However, we will see that these candidates can still be considered approximate clones, as they generally have very low values of $\alpha$ and $\beta$. On average, the minimum values of $\alpha$ and $\beta$ over all profiles are respectively $0.24$ and $0.44$. However, these low values are mostly driven by pairs of candidates from the same party, which are much more likely to be ranked close to each other (e.g., due to party instructions). For instance, 99.4\% of the 664 pairs of candidates that are $\alpha$-deletion clones for $\alpha \le 0.2$ are from the same party, and if we consider only pairs of candidates from distinct parties in each profile, the average minimum value of $\alpha$ and $\beta$ respectively increase to $0.45$ and $0.9$. This is also clear in \Cref{fig:scotland_1}, where each dot represent a pair of candidates in one election, and red dots correspond to pairs of candidates from the same party. \Cref{fig:scotland_1} also highlights a strong correlation between the values of $\alpha$ and $\beta$ in this dataset, especially for pairs of candidates that are the closest to being clones.

Another interesting observation is that the average minimal values of $\alpha$ (and $\beta$) do not seem to depend much on the number of candidates in the election, while the theoretical higher bounds on these values given by \Cref{prop:average-alpha-beta} ($(m-2)/m$ and $(m-2)/3$ respectively) increase with $m$, as shown in \Cref{fig:scotland_2a} (for $\alpha$-deletion clones). This is again mainly due to the presence of candidates from the same party, which are likely to be ranked close to each other regardless of the number of candidates. If we only consider pairs of candidates from distinct parties, we observe a slight increase of the average minimum value of $\alpha$ (and $\beta$) with the number of candidates (see \Cref{fig:scotland_2b} for $\alpha$-deletion clones).

Finally, the mini-jury deliberation dataset contains many perfect clones, with $36.5\%$ of profiles (943 out of 2\,581) admitting perfect clones, and $47.3\%$ of them admitting $1/n$-swap clones (i.e., they are \emph{one swap away} of having perfect clones). One reason is that the number of candidates and voters are both very low, but it is not the only reason: if we randomly sample profiles with 4 candidates and 5 voters using the impartial culture (IC) model (in which every ranking has the same probability to be sampled for every voter) we obtain perfect clones in only $17.3\%$ of profiles (on average over 1\,000\,000 samples). Another reason is that many of the profiles are close to (or equal to) the identity profile in which all voters have the same ranking.  


\begin{figure}[!t]
\centering
    \begin{subfigure}{.49\textwidth}
\centering    
\includegraphics[width=1\textwidth]{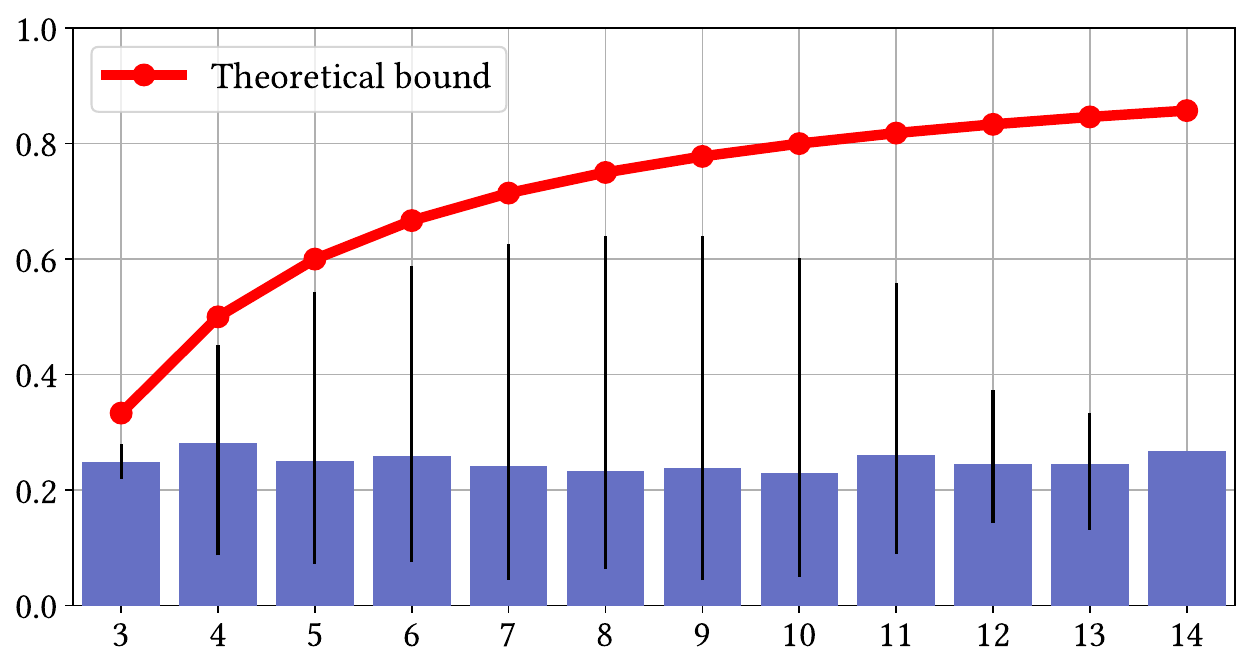}
\caption{Over \emph{all} pairs of candidates in each profile.}
\label{fig:scotland_2a}
\end{subfigure}
    \begin{subfigure}{.49\textwidth}
\centering    
\includegraphics[width=1\textwidth]{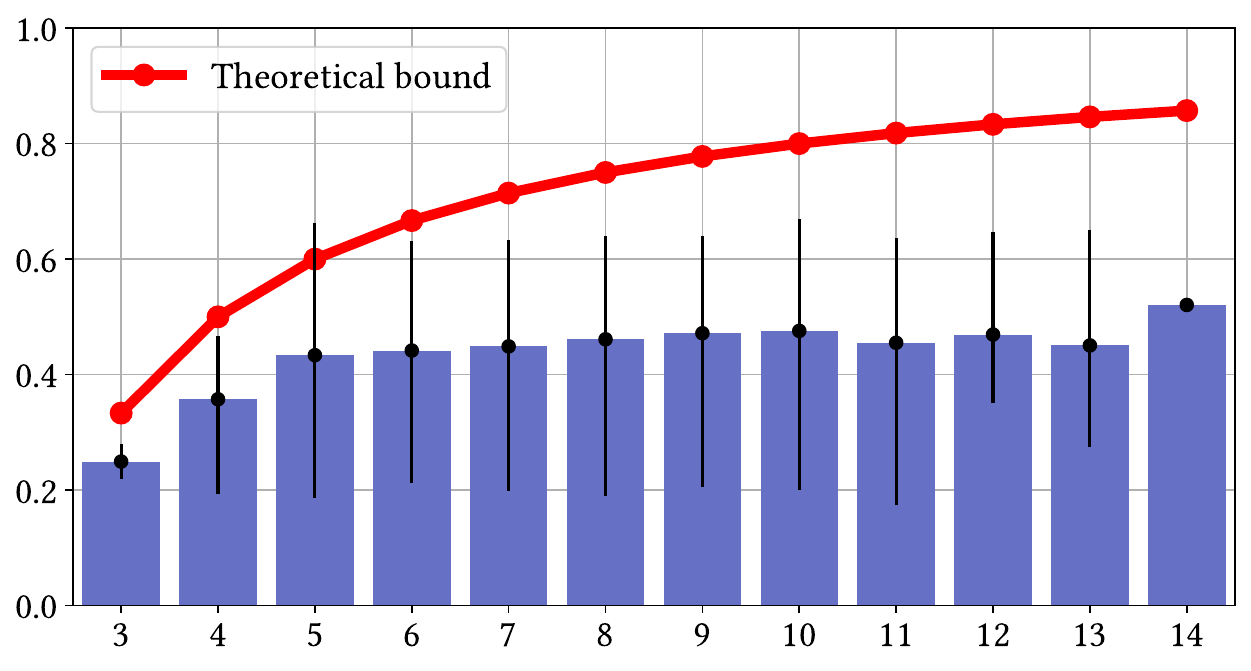}
\caption{Over pairs of candidates from \emph{distinct} parties.}
\label{fig:scotland_2b}
\end{subfigure}

\caption{Blue bars indicate the average minimum value of $\alpha$ for which two candidates are $\alpha$-deletion clones over all instances in the Scottish elections dataset, depending on the number of candidates $m$. The red line indicates the theoretical maximum $(m-2)/m$, and the black lines show the extreme values over all profiles for a given $m$.}
\label{fig:scotland_2}
\end{figure}

\begin{figure}[!t]
\centering
\includegraphics[width=0.49\textwidth]{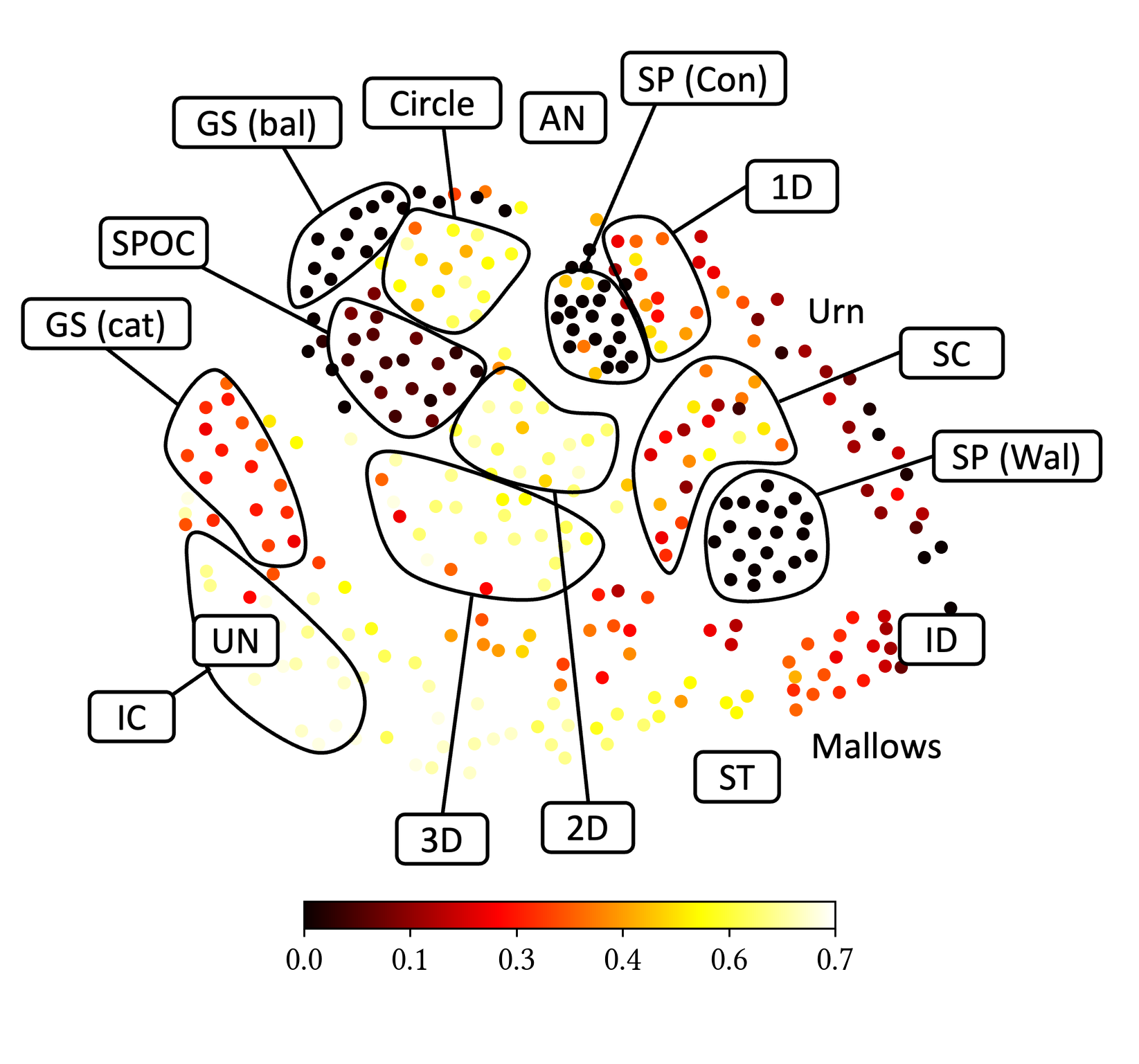}
\hfill 
\includegraphics[width=0.49\textwidth]{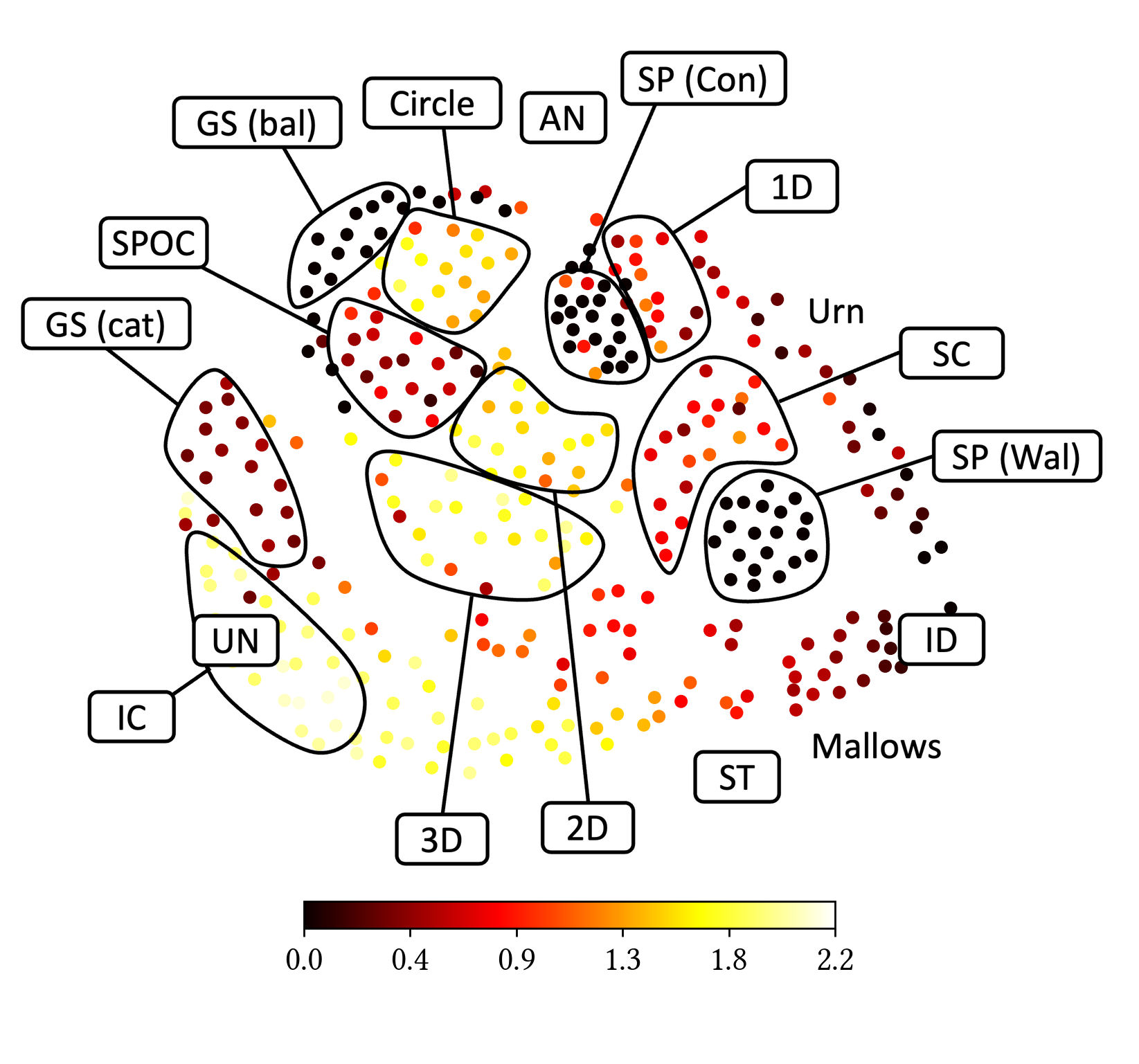}
\caption{Minimum value of $\alpha$ (left) and $\beta$ (right) for which two candidates are respectively $\alpha$-deletion and $\beta$-swap clone (over all pairs of candidates) for profiles with $m=10$ candidates and $n=50$ voters, visualized on the \emph{map of elections} \citep{szufa2020drawing}. Each point corresponds to a profile.}
\label{fig:map}
\end{figure}


We complete this analysis by computing the minimum value of $\alpha$ for which two candidates are $\alpha$-deletion clones and $\beta$-swap clones in profiles sampled from various models using the framework of the \emph{map of elections} \citep{szufa2020drawing,boehmer2021putting}. On a map of elections, each point corresponds to a preference profile (in this case with $m=10$ candidates and $n=50$ voters), and the positions of the points are based on some distance between the profiles and some embedding algorithm. Intuitively, points that are close to each other on the map represent preference profiles with similar structures. In \Cref{fig:map}, we used the swap distance and the embedding provided by the Kamada-Kawai algorithm \citep{kamada1989algorithm} to obtain the positions of the profiles, and we computed for each profile the minimum value of $\alpha$ (resp. $\beta$) for which two candidates are $\alpha$-deletion clones (resp. $\beta$-swap clones). Note that this particular experiment has also been conducted by \citet[Figure 1]{janeczko2024discovering}. 

We observe that for the impartial culture (IC), we obtain high values of $\alpha$ and $\beta$, while we obtain low values for profiles close to the identity profile (ID), confirming our previous results. It appears that the Urn model (another popular model in social choice), gives elections with low minimum values of $\alpha$ and $\beta$ as well. This is also the case of highly structured preferences (in particular single-peaked and single-crossing preferences, but also 1D-Euclidean preferences). Interestingly, we do not obtain such low values of $\alpha$ and $\beta$ for Euclidean preferences in higher dimensions (e.g., 2D, 3D or Circle). 


\subsection{Independence of Approximate Clones}

Now that we have established the presence of approximate clones in these datasets, we can investigate whether voting rules are independent of approximate clones in practice. This will also give us some insights on the strength of the different independence axioms.

We will include in our analysis four voting rules: in addition to IRV and Ranked Pairs, that are independent of clones, we consider Plurality and Borda, which are not. We omit Schulze for this empirical analysis and we expect it to behave quite similarly to Ranked Pairs. Plurality selects the candidate(s) which are ranked first by the highest number of voters, and Borda selects the candidate(s) with the highest Borda score, where the Borda score of some candidate $c \in C$ is the sum of $m-\sigma_i(c)$ over all voters $i \in V$. 

For each dataset, we looked at different metrics:
\begin{itemize}
    \item The proportion of pairs of \emph{perfect} clones for which the rules break independence of \emph{perfect clones} (out of 34 pairs for the figure skating dataset and 1\,240 for the deliberation dataset).
    \item The proportion of pairs of $\alpha$-deletion clones for $0 < \alpha \le 0.2$ for which the rules break (weak) independence of approximate clones (out of 48 pairs for the figure skating dataset, 2\,751 pairs for the deliberation detaset, and 694 pairs for the Scottish elections dataset).
    \item The proportion of pairs of candidates(independently of $\alpha$) for which the rules break (weak) independence of approximate clones.
    \item The proportion of profiles for which the rules break independence of \emph{losers}, meaning that removing any of the losing candidates can change the outcome of the election. 
\end{itemize}

\begin{table*}[!t]
\centering
\begin{tabular}{l l c c c c} \toprule
    Dataset & Rule & Perfect clones & Approx. clones ($\alpha \le 0.2$) &  All pairs & Losers \\ \midrule 
    \multirow{4}{*}{Deliberation} & IRV & 1 & 0.85 (0.96) & 0.62 (0.96) & 0.90 \\
    & Ranked Pairs & 1 & 0.83 (0.96) & 0.62 (0.97) & 0.95 \\
    & Borda & 0.91 & 0.71 (0.93) & 0.54 (0.92) & 0.78 \\
    & Plurality & 0.91 & 0.73 (0.92) & 0.48 (0.90) & 0.74 \\\midrule
    \multirow{4}{*}{Figure Skating} & IRV & 1& 1 (1) & 0.92 (1) & 1 \\
    & Ranked Pairs &1 &  1 (1) & 0.92 (1) & 1 \\
    & Borda & 1 &  1 (1) & 0.91 (0.99) & 0.96 \\
    & Plurality & 1 &  1 (1) & 0.91 (0.99) & 0.96 \\\midrule 
    \multirow{4}{*}{Scottish elections} & IRV & - &0.94 (1) & 0.76 (0.99) & 0.93 \\
    & Ranked Pairs & - & 0.94 (1) & 0.78 (0.99) & 0.98\\
    & Borda & - & 0.84 (0.92) & 0.72 (0.98) & 0.84 \\
    & Plurality & - & 0.84 (0.87) & 0.66 (0.96) & 0.68 \\\bottomrule 
\end{tabular}
\caption{Proportion of pairs of candidates/profiles for which the rules break independence axioms. Numbers in parentheses indicate the proportions for the weak version of the independence of approximate clones axiom.}
\label{tab:independence}
\end{table*}

The results are summarized in \Cref{tab:independence}. We clearly see that IRV and Ranked Pairs are much more robust and satisfy independence axioms in more cases than Borda and Plurality. In particular, they more often satisfy independence of losers. They also satisfy independence of approximate clones much more often for pairs of candidates that are close to being clones (i.e., with $\alpha \le 0.2$). One exception is the figure skating dataset, in which all rules perform well, as there are generally clear winners (thus removing a candidate rarely changes the outcome of the election).

\begin{figure}[!t]
\centering
\begin{subfigure}{.49\textwidth}
    \centering
\includegraphics[width=1\textwidth]{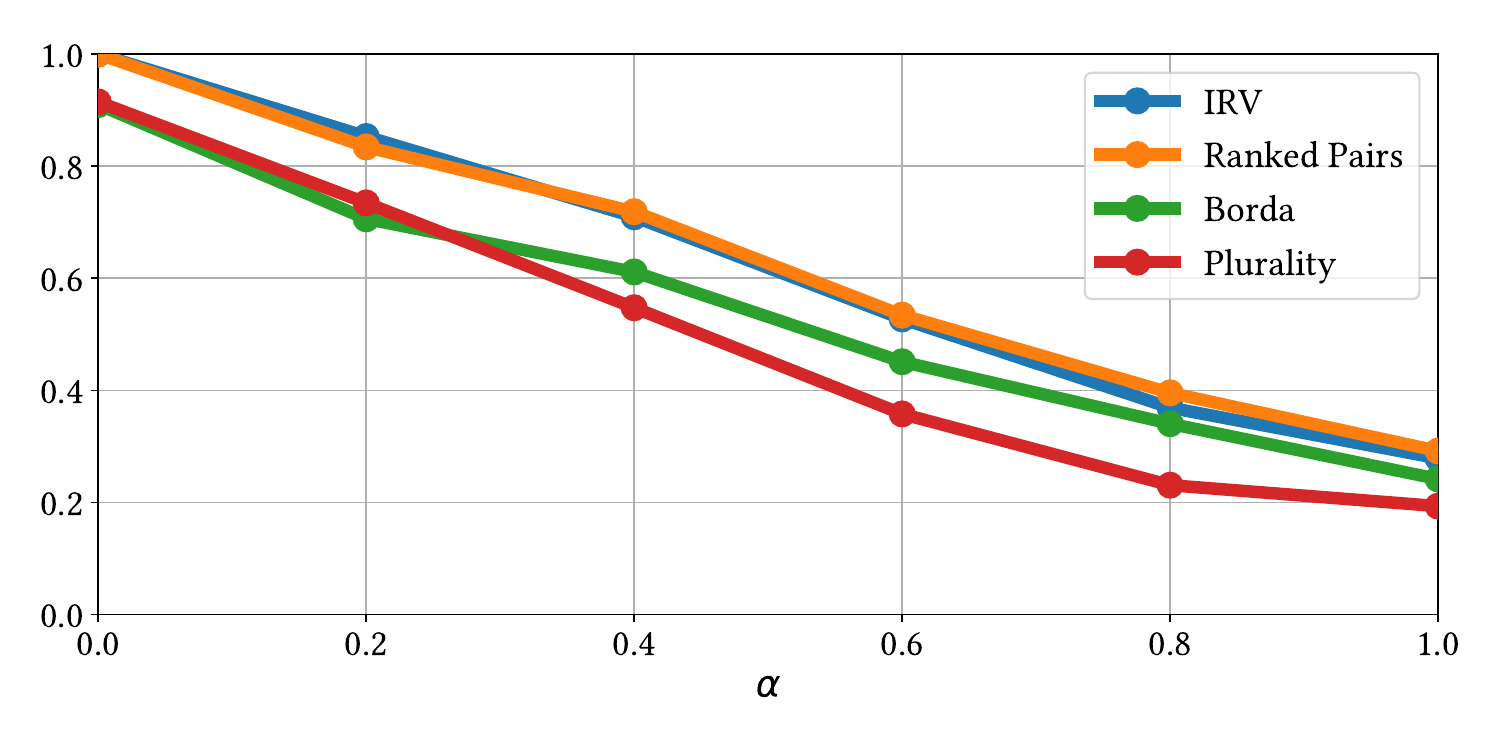}  
\caption{Independence of $\alpha$-deletion clones.}  
\label{fig:habermasa}
\end{subfigure}
\hfill
\begin{subfigure}{.49\textwidth}
    \centering
\includegraphics[width=1\textwidth]{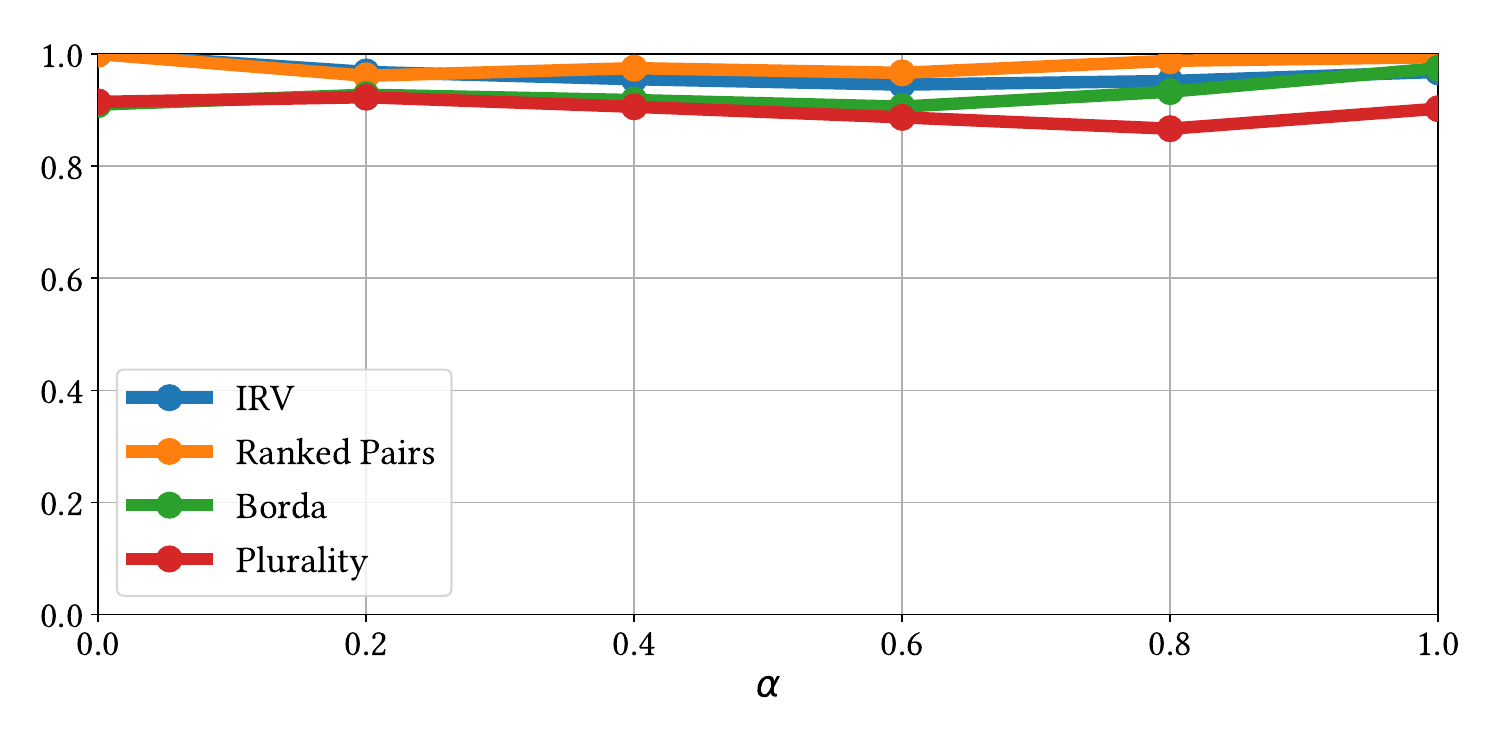}  
\caption{Weak independence of $\alpha$-deletion clones.}  
\label{fig:habermasb}
\end{subfigure}
\caption{Proportion of pairs of candidates for which each rule satisfies (weak) independence of $\alpha$-deletion clones on the Habermas dataset, for different values of $\alpha$.}
\label{fig:habermas}
\end{figure}

A surprising observation is that in all datasets, all the rules we consider satisfy weak independence of approximate clones for more than $90\%$ of pairs of candidates, independently of the value of $\alpha$ for which they are $\alpha$-deletion clones. More surprisingly, the rules break weak independence of approximate clones less often for pairs with very high values of $\alpha$ (the pairs furthest from being clones) than for pairs with moderate values of $\alpha$. On the deliberation dataset for instance, weak independence of approximate clones is broken by IRV for only $3\%$ of the 1\,220 pairs with $\alpha = 1$,  while it is broken for $5\%$ of the 3\,716 pairs with $\alpha = 0.4$ (see \Cref{fig:habermasb}). This counter-intuitive behavior can be explained by the fact that in this dataset, when two candidates are very far from being clones, then very often this means one candidate is always ranked among the last in voters' rankings and the other one always ranked among the first in voters' rankings, making it easy to satisfy the axiom by removing the clear loser. This suggests that weak independence of clones, despite being hard to satisfy in theory, might be too weak of an axiom to be really meaningful in practice.

However, the behavior of rules regarding the strong version of independence of approximate clones is much more insightful. Indeed, as it can be seen in \Cref{fig:habermasa}, the closer two candidates are to being clones (i.e., the lower the value of $\alpha$), the more likely it is that the voting rules satisfy independence of $\alpha$-deletion clones for these candidates. This figure also highlights again that IRV and Ranked Pairs are much more robust than Plurality and Borda, especially for low values of $\alpha$.

\section{Discussion and Further Remarks} \label{sec:conclu}

We analyzed two measures, $\alpha$-deletion clones and $\beta$-swap clones, to identify and quantify the proximity of approximate clones in ordinal preference profiles. Our theoretical analysis showed that results on independence of perfect clones that hold for some voting rules such as IRV or Ranked Pairs generally do not extend to approximate clones. A key open question remains whether any sensible voting rule can satisfy independence of approximate clones for non-trivial thresholds of $\alpha$ or $\beta$, or if  such theoretical robustness is fundamentally unattainable. 

This finding challenges the practical relevance of the standard independence of clones axiom, suggesting that its guarantees do not robustly extend to the more realistic scenario of approximate clones  (as we observed, perfect clones do not appear in actual large-scale political elections, even among same-party candidates). On the other hand, our empirical analysis of real-world datasets showed that approximate clones exist in many contexts in practice, and we observed that despite not satisfying (weak) independence of approximate clones in theory, IRV and Ranked Pairs still satisfy it often in practice, and the closer two candidates are to being clones, the more likely it is that these rules will satisfy it. 

In this paper, we attempted to escape the impossibility result given by \Cref{thm:indep-alpha-clones-imp} by weakening the independence of approximate clones axiom, but we saw that the weaker axiom appeared to be too weak to be meaningful in practice. Another interesting way to escape this impossibility would be to consider randomized voting rules, and to define a weakening of independence of approximate clones where the probability of winning of a candidate can only change by a function of $\alpha$ (or $\beta$) when an approximate clone is removed from the election. This could be an interesting direction for future work.

Future work could also extend the model of approximate clones by focusing on sets of clones instead of pairs of clones, using for instance the notions introduced by \citet{janeczko2024discovering} or \citet{faliszewski2025identifying}. Note that our measures can be generalized to bigger sets in several additional ways, leading to potential interesting axiomatic and empirical comparison of the different generalizations. Finally, a similar analysis of approximate clones could be conducted for other preference formats, such as approval or cardinal preferences. 

\section*{Acknowledgments}
We thank Ratip Emin Berker for the helpful discussions and comments on an earlier version of this paper (in particular for pointing out the counterexample of \Cref{thm:irv-ranked-pairs-alpha-clones} when the profile is not $f$-simple), as well as anonymous reviewers from AAMAS and MPREF for their valuable feedback. This work is supported by the ERC Sinergy Grant ADDI, ID: 101166894, \url{http://doi.org/10.3030/101166894}.

\begin{center}
\includegraphics[width=0.4\textwidth]{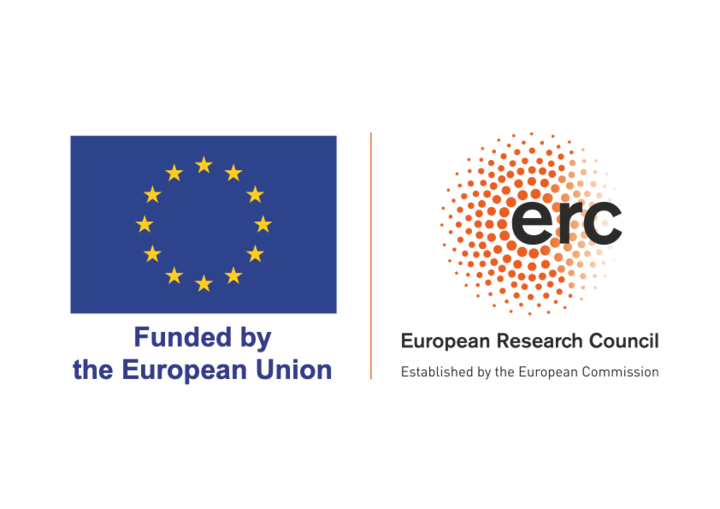}
\end{center}

\bibliographystyle{ACM-Reference-Format}
\bibliography{mybibfile}

\end{document}